\documentclass[12pt]{article}
\usepackage[margin=2cm]{geometry}
\usepackage{amsmath,amstext,amsfonts,amsthm,amssymb,bbm,hyperref}

\def\Sp{\operatorname{Sp}}
\def\Fast{\operatorname{Fast}}

\def\Reg{\operatorname{Reg}}

\theoremstyle{plain}
\newtheorem{thm}{Theorem}
\newtheorem{defin}{Definition}[section]
\newtheorem{assum}[defin]{Assumption}
\newtheorem{prop}[defin]{Proposition}
\newtheorem{cl}[defin]{Claim}
\newtheorem{rmk}[defin]{Remark}
\newtheorem{cor}[defin]{Corollary}
\newtheorem{lemma}[defin]{Lemma}
\newtheorem*{thm*}{Theorem}

\begin{document}

\title{Lower bounds on Anderson-localised eigenfunctions on a strip}
\author{Ilya Goldsheid and Sasha Sodin\footnote{School of Mathematical Sciences, Queen Mary University of London, London E1 4NS, United Kingdom. Email: \{i.goldsheid, a.sodin\}@qmul.ac.uk. The second author is supported in part by the European Research Council starting grant 639305 (SPECTRUM) and by a Royal Society Wolfson Research Merit Award.}}
\maketitle

\begin{abstract} It is known that the eigenfunctions of a random Schr\"odinger operator on a strip decay exponentially, and that the rate of decay is not slower than prescribed by the slowest Lyapunov exponent. A variery of heuristic arguments suggest that no eigenfunction can decay faster than at this rate. We make a step towards this conjecture (in the case when the distribution of the potential is regular enough) by showing that, for each eigenfunction, the rate of exponential decay along any subsequence is strictly slower than the fastest Lyapunov exponent, and that there exists a subsequence along which it is equal to the slowest Lyapunov exponent.
\end{abstract}

\section{Introduction}
Let $W \geq 1$, and let $V(n)$, $n \geq 0$, be independent, identically distributed  random variables taking values in the space of $W \times W$ real symmetric matrices, so that
\begin{equation}\label{eq:model} \mathbb E \| V(n)\|^\eta   < \infty \quad \text{for some $\eta > 0$,}\end{equation}
and the support $\mathcal S$ of the distribution of $V(n)$ is sufficiently rich, say, in the following sense:
 \begin{equation}\label{eq:jcond}\left[\begin{split}
&\text{$\mathcal S$ is irreducible (i.e.\ does not preserve any non-trivial linear subspace of $\mathbb R^W$)}\\
&\text{and contains $V, V'$ such that $\operatorname{rk}(V - V') = 1$}\end{split}\right.
\end{equation}
The main example (the Schr\"odinger case) is
\begin{equation}\label{eq:strip} V(n)_{\alpha,\alpha'} = \begin{cases}
1~, &|\alpha-\alpha'| = 1 \\
v_{n,\alpha}~, &\alpha = \alpha'
\end{cases} \qquad \text{(and $0$ otherwise),}\end{equation}
where $\{v_{n,\alpha}\}_{n \in \mathbb Z_+, \alpha \in \{1,\cdots,q\}}$ are independent, identically distributed real-valued random variables not concentrated at one point and having $\mathbb E  |v_{n,\alpha}|^\eta<\infty$.

We are interested in the spectral properties of the random
operator $H$ on $\ell_2(\mathbb Z_+ \to \mathbb C^W)$, defined as follows:
\[ (H \psi)(n) = \begin{cases}
\psi(n+1) + V(n) \psi(n) + \psi(n-1)~, & n \geq 1\\
\psi(1) + V(0) \psi(0)~, &n = 0~.
\end{cases}\]
On an event of full probability, $H$ exhibits Anderson localisation which manifests itself in the following spectral properties: the spectrum of $H$ is pure point, and the eigenfunctions decay exponentially, meaning that there exists a deterministic $\gamma > 0$ such that for each eigenfunction $\psi$ of $H$
\begin{equation}\label{eq:expdecay} \limsup_{n \to \infty} \frac{1}{n} \log \|\psi(n)\| \leq - \gamma~. \end{equation}
where $\| \cdot \|$ denotes the Euclidean norm in $\mathbb R^W$.

For $W=1$, the pure point nature of the spectrum was first established in \cite{GMP}, and exponential decay -- by Molchanov in \cite{Molch}; see further Kunz and Souillard \cite{KS}. In these works, it was assumed that the distribution of the potential is absolutely continuous with bounded density, The case of singular potentials was settled by Carmona, Klein, and Martinelli \cite{CKM}. For $W>1$ (Schr\"odinger case) with absolutely continuous distribution of the potential, the pure point nature of the spectrum  was first proved in \cite{G80}, and  exponential decay -- by Lacroix in  \cite{Lacr,Lacr2}. The general Schr\"odinger case was settled by Klein, Lacroix and Speis in \cite{KLS}, building on \cite{GM}. The argument of \cite{KLS}  can be extended to the general situation (\ref{eq:jcond}), once the result of \cite{G95} (discussed below) is taken into account; an alternative argument avoiding multi-scale analysis and applicable to the general model (\ref{eq:model}) (and also to its further generalisation allowing for random hopping) is given in \cite{MS}. In this paper, we do not discuss Anderson localisation in dimension $d > 1$, and refer to the works of Fr\"ohlich and Spencer \cite{FS} and Aizenman and Molchanov \cite{AM} and also to the monograph of Aizenman and Warzel \cite{AW}.

\medskip A more precise version of the relation (\ref{eq:expdecay}) can be stated in terms of the Lyapunov exponents associated with $H$. For  $\lambda \in \mathbb R$, define the one-step transfer matrices
\[ T_n (\lambda) = \left( \begin{array}{cc} \lambda - V(n) & - \mathbbm 1 \\ \mathbbm 1 & 0 \end{array} \right) \in \Sp(2W, \mathbb R) \qquad (n \geq 0) \]
and  the multi-step transfer matrices
\[ \Phi_{n, n'}(\lambda) =  T_{n-1}(\lambda) \cdots T_{n'}(\lambda)~, \quad \Phi_n(\lambda) = \Phi_{n,0}(\lambda) \qquad (n > n' \geq 0)~. \]
The Lyapunov exponents $\gamma_1(\lambda) \geq \gamma_2(\lambda) \geq \cdots \geq \gamma_{2W}(\lambda)$ are defined as
\[ \gamma_j(\lambda) = \lim_{n \to \infty} \frac1n \mathbb E \log s_j(\Phi_n(\lambda))~, \]
where $s_j$ stands for the $j$-th singular value. According to a  general result of Furstenberg and Kesten \cite{FK}, one
 has
\begin{equation}\label{eq:fk}
\forall \lambda \in \mathbb R \,\,\,\,\, \mathbb P \left\{ \gamma_j(\lambda) = \lim_{n \to \infty} \frac1n  \log s_j(\Phi_n(\lambda)) \right\} = 1~.\end{equation}
Due to the symplectic structure, $\gamma_{2W+1-j}(\lambda) = - \gamma_j(\lambda)$ for $j =1, \cdots, W$.

Following precursory work by Tutubalin (see the survey \cite{SazTut}) and Virtser \cite{Vir}, Guivarc$'$h and Raugi showed \cite{GR} that if
\begin{equation}\label{eq:cond}
\left[ \begin{split}
&\text{the action of the semigroup generated by the support $\mathcal S_\lambda$ of $T_n(\lambda)$ } \\
&\text{on $\mathbb R^{2W}$ and its wedge powers is strongly irreducible and contractive,}
\end{split} \right.
\end{equation}
 then the Lyapunov exponents are distinct:
\begin{equation}\label{eq:simplespec} \gamma_1(\lambda) > \gamma_2(\lambda) \cdots >\gamma_W(\lambda) > 0~. \end{equation}
In the case (\ref{eq:strip}) with absolutely continuous distribution of $v_{n,\alpha}$, the condition (\ref{eq:cond}) was verified in \cite{Lacr}, while in \cite{G80} (\ref{eq:simplespec}) was directly established using the results of \cite{SazTut}. In \cite{GM}, the following general theorem is proved: (\ref{eq:cond}) (and consequently also (\ref{eq:simplespec})) holds  if
\begin{equation}\label{eq:zcond}
\text{the group generated by $\mathcal S_\lambda$ is Zariski-dense in $\Sp(2W, \mathbb R)$}.
\end{equation} It was also shown in \cite{GM}  that in the Schr\"odinger case (\ref{eq:strip}) one has (\ref{eq:zcond}) for any $\lambda \in \mathbb R$. In \cite{G95}, a general method to compute the Zariski closure of the group generated by the support of  $T_n(\lambda)$  was developed; one of its consequences is that (\ref{eq:zcond}) holds for any $\lambda \in \mathbb R$ also in the generality of (\ref{eq:jcond}).

Now we can state the full result of Klein, Lacroix and Speis \cite{KLS} (in the current setting, covered by \cite{MS}): there is an event of full probability on which each eigenpair $H\psi = \lambda \psi$ satisfies
\begin{equation}\label{eq:upperbd}  \limsup_{n \to \infty} \frac1n \log \|\psi(n)\|  \leq - \gamma_W(\lambda)~. \end{equation}
A variety of heuristic arguments indicate that (\ref{eq:upperbd}) should be sharp in the following  strong sense: there is  an event of full probability on which  each eigenpair $H\psi = \lambda \psi$ satisfies
\begin{equation}\label{eq:conjlowerbd}  \text{(conjecture)}\qquad  \liminf_{n \to \infty} \frac1n \log (\| \psi(n)\| + \|\psi(n+1)\|)  \geq - \gamma_W(\lambda)~, \end{equation}
which, in conjuction with (\ref{eq:upperbd}), implies the existence of a limit equal to $-\gamma_W(\lambda)$.
For example,  the Fermi Golden Rule leads one to believe that eigenfunctions violating (\ref{eq:conjlowerbd}) are unstable under perturbation. From the point of view of random matrix products, an eigenfunction decaying at a rate faster than $\gamma_W$ indicates a non-generic intersection between the $W$-dimensional space of initial conditions with the $W$-dimensional Oseledec subspace of decaying solutions in $\mathbb R^{2W}$.

The relation (\ref{eq:conjlowerbd}) was repeatedly conjectured at least since the 1980s; however, we are not aware of any rigorous results improving on the trivial bound
\begin{equation}\label{eq:trivlowerbd}\liminf_{n \to \infty} \frac1n \log (\| \psi(n)\| + \|\psi(n+1)\|) \geq - \gamma_1(\lambda)  \end{equation}
(which follows from a general result of Craig and Simon \cite{CS}, or from its quantitative version, stated as Lemma~\ref{l:upperbd} below). The main difficulty comes from the fact that, although for each fixed energy $\lambda$ the probability to have an  eigenfunction which decays at a rate faster than $\gamma_W(\lambda)$ is zero, one can not use the union bound over the uncountable set of all real $\lambda$.

\medskip In this paper we make a step towards (\ref{eq:conjlowerbd}) by improving upon (\ref{eq:trivlowerbd}) (in the case when the distribution of potential is regular enough). To state the results precisely, we introduce some notation. Let $\mathcal E(H) = \{ (\lambda, \psi) \}$ be the collection of eigenpairs of $H$, with the normalisation $\|\psi(0)\|=1$ (the choice of the sign is not important for us, and spectral multiplicity is known to be a null event). For $\gamma > 0$ and a bounded interval $I \Subset \mathbb R$, consider the two realisation-dependent sets:
\begin{equation}\begin{split}
\Fast^+(\gamma; I) &= \left\{ \lambda \in I: \exists  (\lambda, \psi) \in \mathcal E(H),\,    \liminf_{n \to \infty}  \frac{\log (\| \psi(n)\| + \|\psi(n+1)\|) }n   \leq - \gamma \right\}~,\\
\Fast^-(\gamma; I) &= \left\{ \lambda \in I: \exists  (\lambda, \psi) \in \mathcal E(H),\,   \limsup_{n \to \infty} \frac{\log (\| \psi(n)\| + \|\psi(n+1)\|) }n   \leq - \gamma \right\} ~.\end{split}\label{eq:deffast}\end{equation}
These sets consist of the eigenvalues for which the corresponding eigenvector decays at rate $\geq \gamma$ (along a subsequence, or uniformly). We note that there is no simple way to define the  sets as random variables on the underlying probability space (see Kendall \cite{Kendall} and Tsirelson \cite{Tsir} for possible frameworks to address such questions); this does not cause problems since we  only work with the measureable events  $\{ \Fast^\pm(\gamma; I) = \varnothing \}$ and $\{ \Fast^\pm(\gamma; I) \neq  \varnothing \}$ (which in fact lie in the tail $\sigma$-algebra).

For $\lambda$ in the spectrum $\sigma(H)$ of $H$, define the deterministic quantities
\begin{equation}\label{eq:degamma}\gamma_*^\pm(\lambda) = \inf \Big\{ \gamma > 0 : \, \exists r > 0~, \, \mathbb P \left\{ \Fast^\pm(\gamma, (\lambda - r, \lambda + r)) \neq \varnothing \right\} = 0 \Big\} ~. \end{equation}
Roughly speaking, the functions $\gamma_*^\pm(\lambda)$ measure  the fastest possible decay of an eigenfunction in the vicinity of $\lambda$ (recall that $\gamma_j(\lambda)$ are continuous, cf.\ below, and that  $\sigma(H)$ is almost surely equal to a deterministic set). In this notation,  (\ref{eq:upperbd}) and (\ref{eq:trivlowerbd})  imply that
\begin{equation}\label{eq:sandwich}\gamma_1(\lambda) \geq \gamma_*^{+}(\lambda)  \geq \gamma_*^{-}(\lambda)  \geq \gamma_W(\lambda)~,\end{equation}
whereas the conjecture (\ref{eq:conjlowerbd}) stipulates  that the last two inequalities are in fact equalities: $\gamma_*^\pm \overset{\text{\tiny conj}}{=} \gamma_
W$. The results below show that the first inequality in (\ref{eq:sandwich}) is strict (for $W \geq 2$), whereas the last one is an equality, at least, if one asumes

\medskip
\begin{assum}\label{assum} (a) The distribution  of $V(n)$ is compactly supported on a real-analytic submanifold $\mathcal M$ in the space of symmetric $W \times W$ matrices, and is absolutely continuous with bounded density with respect to the $(\dim \mathcal M)$-dimensional Lebesgue measure on $\mathcal M$; (b) for each $\lambda \in \mathbb R$ the image of $\mathcal M$ under
\[ V \mapsto  \left( \begin{array}{cc} \lambda - V & - \mathbbm 1 \\ \mathbbm 1 & 0 \end{array} \right) \]
generates $\Sp(2W, \mathbb R)$ as a Lie group.
\end{assum}

\begin{rmk} Assumption~\ref{assum} implies both  (\ref{eq:model}) and  (\ref{eq:zcond}).
\end{rmk}

\begin{rmk}In the Schr\"odinger case (\ref{eq:strip}), Assumption~\ref{assum} is satisfied if the random variables $v_{n,\alpha}$ are bounded and their distribution is absolutely contiunous with bounded density (see \cite[Section 1.4]{Lacr3}).\end{rmk}

\begin{thm}\label{thm:1}
Let $W \geq 3$. If Assumption~\ref{assum} holds, then $\gamma_*^{+}(\lambda) \leq \gamma_{*,1}(\lambda)$ for $\lambda \in \sigma(H)$, where $\gamma_{*,1}(\lambda)$ is the unique solution of the equation
\begin{equation}\label{eq:thm1} \big( (W-1) \gamma - \sum_{j=1}^{W-1}\gamma_j(\lambda)\big)_+ + \gamma = \gamma_1(\lambda)~. \end{equation}
\end{thm}
Here $x_+ = \max(x, 0)$. We observe that for $W \geq 3$ $\gamma_{*,1}(\lambda) < \gamma_1(\lambda)$, hence (\ref{eq:thm1}) indeed improves on (\ref{eq:trivlowerbd}).  For $W = 2$, $\gamma_{*,1} = \gamma_1(\lambda)$; however, we prove
\begin{thm}\label{thm:2}
Let $W  = 2$. If Assumption~\ref{assum} holds, then $\gamma_*^{+}(\lambda) \leq  \frac{2 \gamma_{1}(\lambda) + \gamma_2(\lambda)}{3}$ for all $\lambda \in \sigma(H)$.
\end{thm}
\noindent As for  $\gamma_*^{-}$, our methods yield the optimal result:
\begin{thm}\label{thm:3}
 Let $W \geq 2$. If Assumption~\ref{assum} holds, then $\gamma_*^{-}(\lambda) = \gamma_{W}(\lambda)$ for all $\lambda \in \sigma(H)$.
\end{thm}
\noindent The following corollary summarises our main conclusions:
\begin{cor} Let $W \geq 2$. If Assumption~\ref{assum} holds, then
\begin{equation}
\gamma_W(\lambda) = \gamma_*^{-}(\lambda) \leq \gamma_*^{+}(\lambda) < \gamma_{1}(\lambda)
\end{equation}
for all $\lambda \in \sigma(H)$.
\end{cor}

In the proofs, we repeatedly use the following argument, inspired by the work of Kakutani \cite{Kakutani} and its ramifications by Spencer and Aizenman \cite{Aiz}, to  estimate the probability of exceptional events. Suppose  we want to bound $\mathbb P\left\{ A \neq \varnothing\right\}$, where  $A$ is a random subset of, say, the interval $[0, 1]$. Suppose we find $\eta \in (0, 1]$ and a random superset $A^{+} \supset A$ with the following properties: (a)~for each $\lambda \in [0, 1]$, $\mathbb P \left\{ \lambda \in A^{+}\right\} \leq p$ (``single-energy bound''); (b) if $\lambda \in A$ and $|\lambda' - \lambda| < \eta$, then $\lambda' \in A^{+}$ (``propagation estimate''). Then  the Chebyshev inequality and the Fubini theorem yield:
\[ \mathbb P \left\{ A \neq \varnothing \right\}
\leq \mathbb P \left\{ \operatorname{mes} (A^+ \cap [0,1]) \geq \eta \right\}
\leq \frac1\eta \mathbb E \operatorname{mes} (A^+ \cap [0,1])
\leq \frac{p}{\eta}~.\]
The paper is organised as follows. Some preliminary estimates are collected in Section~\ref{s:prelim}.  In Sections~\ref{s:1} and \ref{s:2} we prove Theorems~\ref{thm:2} and \ref{thm:3}, respectively. In Section~\ref{s:concl} we discuss the prospects of improving the bounds in Theorems~\ref{thm:1} and \ref{thm:2}, and point out the connection to the problem, going back to \cite{G75,G80a} and recently studied by Gorodetski and Kleptsyn \cite{GorKl}, of uniform convergence to the Lyapunov exponents, i.e.\ whether  the quantifier $\forall \lambda$ in (\ref{eq:fk}) can be inserted inside the curly brackets. We also prove Proposition~\ref{p:gorkl}, which is an $\Sp(2W, \mathbb R)$-counterpart of one of the results of \cite{GorKl}. The proof of  Theorem~\ref{thm:3} in Section~\ref{s:3} makes use of this proposition.

\medskip
We conclude this introduction with two remarks. First, we have chosen to present the arguments for the one-sided strip $\mathbb Z_+ \times \{1, \cdots, W\}$; similar arguments can be applied to the two-sided strip $\mathbb Z \times \{1, \cdots, W\}$. Second, it is possible that Assumption~\ref{assum} can be somewhat relaxed, and that  a refinement of the current methods could be applicable when the invariant measure (describing the limiting distribution of the unitary matrices in the singular value decomposition of the transfer matrices $\Phi_n$) is absolutely continuous with bounded density with respect to the Haar measure on the compact symplectic group, or at least enjoys the Frostman property (upper bound on the measure of every ball by a power of the radius) with a sufficiently large exponent. On the other hand, it is known (see \cite{H} for the case $W=1$) that for singular distributions of $V(n)$ the invariant measure may be supported on lower-dimensional subsets of the symplectic group. Extending our results to such cases would require additional  ideas.


\section{Preliminaries}\label{s:prelim}

\paragraph{Convergence to the Lyapunov exponent} Assume that (\ref{eq:cond}) holds at some $\lambda$. Then (\ref{eq:cond}) also holds in a neighbourhood of $\lambda$, and then (see e.g.\ \cite[Corollary 2.5]{KLS})  the Lyapunov exponents $\gamma_j(\lambda)$ are continuous at $\lambda$. For each $\epsilon > 0$, let $r_\epsilon(\lambda)\in (0, 1/2]$ be such that
\begin{equation}\label{eq:def-r} \forall \lambda' \in (\lambda - r_\epsilon(\lambda), \lambda + r_\epsilon(\lambda)) \,\, \forall 1 \leq j \leq W \,\,\,\, |\gamma_j(\lambda') - \gamma_j(\lambda)| < \epsilon~. \end{equation}
The following large deviation estimate goes back to the work of Le Page \cite{LP}.
\begin{lemma}[see  \cite{DK}, {\cite[Section V.6]{BougL}}]\label{l:largedev}
Assume (\ref{eq:model}). Let $I \Subset \mathbb R$  be a finite interval such that (\ref{eq:cond}) holds for all $\lambda \in I$. Then there exist $C>0$ and $c>0$ such that for each $\lambda \in I$,  $1 \leq j \leq W$, $\epsilon \in(0, 1]$, and $n \geq 1$
\begin{equation}\label{eq:largedev} \mathbb P\left\{\left| \frac{1}{n} \log s_j(\Phi_n(\lambda)) - \gamma_j(\lambda) \right| \geq \epsilon\right\} \leq C \exp(-c \epsilon^2 n)~.\end{equation}
\end{lemma}

\smallskip
The arguments leading to the following corollary of Lemma~\ref{l:largedev} are also well known (for $W = 1$, see e.g.\ Jitomirskaya and Zhu \cite[Section 5]{JZ}; we also mention a result of Craig--Simon \cite[Theorem 2.3]{CS}, which is not quantitative, but on the other hand holds in more general setting).
\begin{lemma}\label{l:upperbd} Assume (\ref{eq:model}). Suppose $\lambda \in \mathbb R$ is such that (\ref{eq:cond}) holds. Then there exist $C>0$ and $c>0$ such that for each $1 \leq j \leq W$, $\epsilon \in(0, 1]$, and $n \geq 1$
\[ \mathbb P\left\{ \exists \lambda' \in  (\lambda - r_\epsilon(\lambda), \lambda + r_\epsilon(\lambda)) \,\, : \,\, \frac{1}{n} \sum_{i=1}^j \log s_i(\Phi_n(\lambda')) \geq \sum_{i=1}^j \gamma_i(\lambda) + 2j \epsilon \right\} \leq C n \exp(-c \epsilon^2 n)~.\]
\end{lemma}
\begin{proof}
If  $n \leq 100 W^2$ or $\epsilon^2 \leq 100 \log n / n$, we can ensure the desired inequality by adjusting the constants, therefore we assume that  $n > 100 W^2$ and $\epsilon^2 \geq 100 \log n / n$.
Consider the $j$-th exterior power $\Phi_n(\lambda')^{\wedge j}$ of $\Phi_n(\lambda')$, so that
\[\log \|\Phi_n(\lambda')^{\wedge j} \| = \sum_{i=1}^j \log s_j(\Phi_n(\lambda'))~.\]
Each matrix element $p(\lambda')$ of $\Phi_n(\lambda')^{\wedge j}$ (where $p$ runs in a finite set $P$ enumerating
the matrix elements) is a polynomial of degree $\leq j n \leq W n$ in $\lambda$. Now we use the following  result of Bernstein \cite{Bern}, although we require much less than its full strength (in place of the logarithmic dependence on the degree with a precise constant, we could do with any prefactor growing slower than exponentially):  for any polynomial $q$ of degree $n$
\[ \max_{|\lambda| \leq 1} |q(\lambda)|  \leq C_n  \max_{\alpha\in \{0,1,\cdots,n\}} |q(\cos (\pi \frac{\alpha+\frac12}{n+1}))|~, \quad \text{where} \quad
 C_n = (1+ o(1)) \, \frac2\pi \log n~.  \]
Returning to our setting, let
\[ \lambda_\alpha = \lambda +  r_\epsilon(\lambda) \cos(\pi \frac{\alpha +\frac12}{Wn+1})~, \quad 0 \leq \alpha \leq Wn~;\]
then we have for any $p \in P$:
\begin{equation}\label{eq:interp}\max_{\lambda' \in  (\lambda - r_\epsilon(\lambda), \lambda + r_\epsilon(\lambda))} |p(\lambda)| \leq C \log (Wn) \max_{\alpha\in \{0,1,\cdots,Wn\}} |p(\lambda_\alpha)| \leq e^{\frac{\epsilon n}{3}} \max_{0 \leq \alpha \leq Wn} |p(\lambda_\alpha)|~. \end{equation}
By Lemma~\ref{l:largedev} and the choice of $r_\epsilon$,
\[ \mathbb P \left\{ |p(\lambda_\alpha)| \geq \exp\left\{n \left[ \sum_{i=1}^j \gamma_i(\lambda) + \frac{4j}{3} \epsilon  \right] \right\}   \right\} \leq C' \exp(-c' \epsilon^2 n)~.\]
Thus by (\ref{eq:interp})
 \[ \mathbb P \left\{ \max_{p \in P}\max_{\lambda' \in  (\lambda - r_\epsilon(\lambda), \lambda + r_\epsilon(\lambda))} |p(\lambda')|   \geq \exp\left\{n \left[ \sum_{i=1}^j \gamma_i(\lambda) + \frac{5j}{3} \epsilon  \right] \right\}   \right\} \leq C'' n  \exp(-c' \epsilon^2 n)~.\]
Finally, $\|\Phi_n(\lambda')\| \leq C  \max_{p \in P} |p(\lambda')| \leq e^{\epsilon n /3} \max_{p} |p(\lambda')|$, and this completes the proof.
\end{proof}

\paragraph{The probability density of transfer matrices} The following lemma builds on the arguments going back to the work of Ricci and Stein \cite{RS}. In the context of random Schr\"odinger operators, similar reasoning appears in the work Shubin, Vakilian and Wolff \cite{SVW}. Recently, a general argument in the setting of motivic  morphisms has been developed by Glazer and Hendel \cite{GH}; further arguments are discussed in \cite{GHS}. For completeness, we sketch a proof (restricted to the generality of the current discussion) below.

\begin{lemma}\label{l:smooth}
Assume Assumption~\ref{assum}. There exists $n_0$ such that the following holds.
\begin{enumerate}
\item[(a)] For any $\lambda \in \mathbb R$ the distribution of $\Phi_{n_0}(\lambda)$ is absolutely continuous with bounded density with respect to the Haar measure on $\Sp(2W, \mathbb R)$.
\item[(b)] Let $\Phi_n(\lambda) = U_n(\lambda) \Sigma_n(\lambda) V_n(\lambda)^*$ be the singular value decomposition of $\Phi_n(\lambda)$. Then there exists $C$ such that for any  $n \geq n_0$ the distributions of $V_n(\lambda)$, $U_n(\lambda)$ and $V_n^*(\lambda) U(\lambda)$ are absolutely continuous with density $\leq C$ with respect to the Haar measure on the compact symplectic group $\Sp(2W, \mathbb R) \cap \operatorname{SO}(2W, \mathbb R)$.
\end{enumerate}
Moreover, the bounds in (a)--(b) are locally uniform in $\lambda$.
\end{lemma}

\begin{rmk} For concreteness, we may assume that the singular value decomposition is constructed so that $\Sigma_n$ is diagonal with strictly decreasing positive entries on the diagonal,  and the first non-zero entry of eich column of $U_n$ and $V_n$ is positive.
\end{rmk}

\begin{proof} Consider the product map
\begin{equation}\label{eq:map} F_{n} = F_{n, \lambda}: \mathcal M^n \to \Sp(2W, \mathbb R)~, \quad (V(1), \cdots, V(n)) \mapsto \Phi_{n}(\lambda)~. \end{equation}
According to \cite[Proposition~1.1]{RS}, for
\[ n_1 = 2^{\dim \Sp(2W, \mathbb R) - \dim \mathcal M} = 2^{W(2W+1) - \dim \mathcal M} \]
the image $F_{n_1}(\mathcal M)$ contains an open set in $\Sp(2W, \mathbb R)$ (in the Schr\"odinger case, the same conclusion holds for $n_1 = \dim \Sp(2W, \mathbb R) \div \dim \mathcal M = 2W+1$; see \cite[Proposition 1.4.35]{Lacr3}). Hence $\det [(D F_{n_1})^* (D F_{n_1})]$ is not identically zero; by continuity, the maximum of its absolute value is bounded away from zero locally uniformly in $\lambda$.

The map (\ref{eq:map}) is real analytic, therefore  the probability density of $\Phi_{n_1}(\lambda)$ lies in $L_p$ for some $p>1$ (this can be proved directly as in \cite{GHS} or deduced from \cite[Proposition~2.1]{RS} using an appropriate embedding theorem), and, again, both $p$ and the bound are locally uniform in $\lambda$.
Applying the inequality
\[ \| f_1 * f_2 * \cdots * f_n \|_\infty \leq \prod_{\alpha=1}^n \|f_\alpha\|_{1+\frac1n}~, \quad f_\alpha \in L_{1 + \frac1n}(\Sp(2W, \mathbb R)) \]
(which is a simple special case of the Young convolution inequality on $\Sp(2W, \mathbb R)$), we obtain that for $n_0=   n_1 (\lfloor (1 - 1/p)^{-1} \rfloor+ 1)$  the density of $\Phi_{n_0}(\lambda)$ is bounded. This proves the first item, from which the second one follows.
\end{proof}

\paragraph{A geometric lemma} Denote by $S(F)$ the unit sphere of an Euclidean vector space $F$. For future reference, we record the following fact (attributed to Archimedes): if $u$ is a random vector uniformly distributed on $S(\mathbb R^\ell)$, then the probability density of the random vector $P_F u$, where $P_F: \mathbb R^\ell \to F$ be the orthogonal projection onto a fixed $k$-dimensional subspace $F \subset \mathbb R^\ell$, $1 \leq k \leq \ell - 1$, is given by
\begin{equation}\label{eq:archimedes} f_{\ell,k}(v) = C_{\ell,k} (1 - \|v\|)_+^{\frac{\ell-k}{2} - 1}~.\end{equation}

\begin{lemma}\label{l:geom} Let $U$ be a random matrix taking values in $\operatorname{SO}(\ell, \mathbb R)$ such that for each $u \in S(\mathbb R^\ell)$ the vector $Uu$ is uniformly distributed on $S(\mathbb R^\ell)$. Let $D = \operatorname{diag}(e^{a_1}, \cdots, e^{a_\ell})$, where $a_1 \geq a_2 \geq \cdots \geq a_\ell$, and let $F \subset \mathbb R^\ell$ be a $k$-dimensional subspace. Then  for any $a_1 \geq a \geq a_\ell$
\[ \mathbb P \left\{ \exists u \in S(F) \, : \, \|D U u \| \leq e^a \right\}
\leq C_\ell \exp \left\{ - \sum_{j=k}^\ell (a_j - a)_+\right\}~.\]
\end{lemma}

\begin{proof}
It is sufficient to prove the estimate for the $\ell_\infty$ norm $\| \cdot \|_\infty$ in place of the Euclidean norm, as this will only affect the value of the numerical constant $C_\ell$. We first observe that for a fixed $u \in S(\mathbb R^\ell)$
\begin{equation}\label{eq:geom-fixed}
\mathbb P \left\{ \|DU u\|_\infty \leq e^{a} \right\} \leq C_\ell \exp(-\sum_{j=1}^\ell (a_j - a)_+)~.
\end{equation}
Indeed, let $j_0$ be such that $a_{j_0} \geq a > a_{j_0+1}$. The random vector $((Uu)_j)_{j=1}^{j_0}$ has  bounded density in a neighbourhood of zero (according to (\ref{eq:archimedes}), for $j_0 \leq \ell- 2$ the density is uniformly bounded, whereas for $j_0 = \ell - 1$ it explodes only on the boundary of the unit ball). Therefore
\[ \mathbb P \left\{ \|DU u\|_\infty \leq e^{a} \right\} =  \mathbb P \left\{ \forall 1 \leq j \leq j_0 \,\, |(DU u)_j| \leq e^{a} \right\} \leq
C_\ell \prod_{j=1}^{j_0} e^{a -a_j} = C_\ell \exp(-\sum_{j=1}^\ell (a_j - a)_+)~,\]
thus concluding the proof of (\ref{eq:geom-fixed}).

Second, we note that if $\|Dv\|_\infty \leq e^a$, then $\|Dv'\|_\infty \leq 2e^a$ for all
\[ v' \in Q_v = \{ v' \in S(\mathbb R^\ell) \, : \, |v_j' - v_j| \leq \exp(- (a_j - a)_+)\}~. \]
For any $k$-dimensional subspace $F_1 \subset \mathbb R^\ell$ and $v \in S(F_1)$, the $k-1$ dimensional measure of the intersection of $Q_v$ with $S(F_1)$ admits the lower bound
\[ \sigma_{k-1} (S(F_1) \cap Q_v)
\geq c_\ell \exp(- \sum_{j=1}^{k-1} (a_j - a)_+)~,\]
whence by  the Chebyshev inequality, the Fubini theorem and (\ref{eq:geom-fixed})
\[\begin{split}
&\mathbb P \left\{ \exists v \in S(UF) \, : \, \|Dv\|_\infty \leq e^a \right\} \\
&\quad\leq \mathbb P \left\{ \sigma_{k-1} \left\{ v' \in S(UF) \, : \, \|Dv'\|_\infty \leq 2e^a\right\} \geq c_\ell \exp(-\sum_{j=1}^{k-1} (a_j - a)_+)\right\} \\
&\quad\leq C_\ell' \exp(\sum_{j=1}^{k-1}(a_j - a)_+) \mathbb E \sigma_{k-1} \left\{ v' \in S(UF) \, : \, \|Dv'\|_\infty \leq 2e^a\right\} \\
&\quad \leq C_\ell'' \exp(\sum_{j=1}^{k-1}(a_j - a)_+) \, \exp(-\sum_{j=1}^\ell (a_j - a)_+)
=C_\ell'' \exp \left\{ - \sum_{j=k}^\ell (a_j - a)_+\right\}~.\qedhere\end{split} \]
\end{proof}


\section{Proof of Theorem~\ref{thm:1}}\label{s:1}

For the whole proof, we fix $\lambda \in \sigma(H)$ and $\gamma > \gamma_{*,1}(\lambda)$. Choose an auxiliary small parameter $\epsilon > 0$; eventually, we shall substitute $\epsilon = \frac{1}{100 W} \min(\gamma - \gamma_{*,1}(\lambda), 1)$.

Denote
\begin{equation}\label{eq:def-omega}
\Omega_{n,\epsilon}(\lambda) = \bigcap_{1 \leq j \leq W} \bigcap _{1 \leq m_1 \leq m_2 \leq n} \Omega_{n,\epsilon}^{m_1,m_2,j}(\lambda)~,
\end{equation}
where
\[ \Omega_{n,\epsilon}^{m_1,m_2,j}(\lambda) = \left\{ \forall \lambda' \in (\lambda - r_\epsilon(\lambda), \lambda + r_\epsilon(\lambda)) \,\,\, \sum_{i=1}^j \log s_i(\Phi_{m_2,m_1}(\lambda')) \leq (m_2 - m_1) \sum_{i=1}^j \gamma_i(\lambda) + 2\epsilon j n \right\}~. \]
From Lemma~\ref{l:upperbd} (and using that $\Phi_{m_2, m_1}$ has the same distribution as $\Phi_{m_2-m_1}$) we obtain the following maximal inequality:
\begin{equation} \label{eq:p-omega}
\mathbb P(\Omega_{n,\epsilon}(\lambda))\geq 1 - C n^3 \exp(-c \epsilon^2 n)~.
\end{equation}
Let
\begin{equation}\label{eq:def-F0} F_0 = \left\{  \binom{v_1}{0} \, : \, v_1 \in \mathbb R^W \right\} \subset \mathbb R^{2W}\end{equation}
be the space of initial conditions. Denote:
\begin{equation}\label{eq:def-fast}
\Fast_{n,\epsilon}(\gamma, \lambda) = \left\{ \lambda' \in (\lambda - r_\epsilon(\lambda), \lambda + r_\epsilon(\lambda)) \,\, : \,\,
\exists v \in S(F_0)~, \,\, \| \Phi_n(\lambda') v\| \leq e^{-n \gamma}\right\}~,\end{equation}
so that for any $\tilde\gamma > \gamma$
\[ \Fast^+\big(\tilde \gamma, (\lambda - r_\epsilon(\lambda), \lambda + r_\epsilon(\lambda))\big) \subset \limsup_{n \to \infty} \Fast_{n,\epsilon}(\gamma, \lambda)~. \]
We shall prove that for sufficiently small $\epsilon$
\begin{equation}\label{eq:est-fast-1}\mathbb{P} \left\{  \Fast_{n,\epsilon}(\gamma, \lambda) \neq \varnothing \right\} \leq C e^{-cn}~; \end{equation}
by the Borel--Cantelli lemma, this estimate will imply that  almost surely
\[ \Fast^+\big(\tilde\gamma, (\lambda - r_\epsilon(\lambda), \lambda + r_\epsilon(\lambda))\big) = \varnothing~, \quad \tilde\gamma > \gamma~,\]
and thus $\gamma_*^+ \leq\gamma$.

\medskip\noindent
The proof of (\ref{eq:est-fast-1}) rests on two claims, a propagation estimate and a single-energy bound. Set $\eta = n^{-1} e^{-n(\gamma + \gamma_1(\lambda) + 4 \epsilon)}$.
\begin{cl}\label{cl:2}
On the event $\Omega_{n,\epsilon}(\lambda)$,
\begin{equation}\label{eq:propag-1}
\left. \begin{split} \lambda', \lambda'' \in (\lambda - r_\epsilon(\lambda), \lambda + r_\epsilon(\lambda))&\\
\lambda' \in \Fast_{n,\epsilon}(\gamma, \lambda)& \\
 |\lambda'' - \lambda'| \leq\eta& \end{split} \right\}  \Longrightarrow \lambda'' \in \Fast_{n,\epsilon}(\gamma - \frac{\log 2}{n}, \lambda)~.
\end{equation}
\end{cl}
\begin{proof} On $\Omega_{n,\epsilon}(\lambda)$, we have
\begin{equation}\label{eq:16.5} \lambda', \lambda'' \in (\lambda - r_\epsilon(\lambda), \lambda + r_\epsilon(\lambda))
\Longrightarrow \| \Phi_n(\lambda') - \Phi_n(\lambda'') \| \leq n e^{n (\gamma_1(\lambda) + 4\epsilon)} |\lambda ' - \lambda''|~, \end{equation}
hence for $|\lambda' - \lambda''| \leq \eta$ we have
\[ \| \Phi_n(\lambda') - \Phi_n(\lambda'')\| \leq e^{-n\gamma}~.\]
If $\lambda' \in \Fast_{n,\epsilon}(\gamma, \lambda)$, then there exists $v \in S(F_0)$ such that $\|\Phi_n(\lambda') v\| \leq e^{-n\gamma}$,
and then
\begin{equation}\label{eq:propag} \|\Phi_n(\lambda'') v\| \leq e^{-n\gamma} + \|\Phi_n(\lambda'')-\Phi_n(\lambda')\| \leq
e^{-n\gamma} + e^{-n\gamma} = 2e^{-n\gamma}~, \end{equation}
i.e.\ $\lambda'' \in \Fast_{n,\epsilon}(\gamma - \frac{\log 2}{n}, \lambda)$, as asserted.\end{proof}

\begin{cl}\label{cl:1} For any $\gamma>0$, $n \geq n_0$, and $\lambda'' \in (\lambda - r_\epsilon(\lambda), \lambda+r_\epsilon(\lambda))$
\begin{equation}\label{eq:singlelam-1}
\begin{split}
&\mathbb P \left\{ \lambda'' \in \Fast_{n,\epsilon}(\gamma, \lambda)~, \, \omega \in \Omega_{n,\epsilon}(\lambda) \right\} \\
&\qquad\leq C \exp\left\{ -n \left[ 2 \gamma + \big( (W-1)\gamma - \sum_{j=1}^{W-1} \gamma_{j}(\lambda) \big)_+ -  2W\epsilon \right]\right\}~.
\end{split}\end{equation}
\end{cl}
\begin{proof}
Let $n_0$ be as in Lemma~\ref{l:smooth}, and let $M \in \Sp(2W, \mathbb R)$ be a random matrix uniformly distributed according to the restriction of the Haar measure to a sufficiently large ball in operator norm. Denote $\widetilde \Phi_n(\lambda) = \Phi_{n, n_0}(\lambda)M$. According to Lemma~\ref{l:smooth}, it suffices to show that
\begin{equation}\label{eq:singlelam-1'}
\mathbb P \left\{ s_W(\widetilde \Phi_n(\lambda'')|_{F_0}) \leq e^{-\gamma n}~, \, \omega \in \Omega_{n,\epsilon}(\lambda)  \right\}
\leq  (\text{RHS of (\ref{eq:singlelam-1})})~.
\end{equation}
Introduce the singular value decompositon
\[ \Phi_{n, n_0}(\lambda'') = U_{n, n_0}(\lambda'') \Sigma_{n, n_0}(\lambda'') V_{n, n_0}(\lambda'')^*~, \quad
M =U \Sigma V^*~,\]
so that
\[\widetilde\Phi_n(\lambda'') = U_{n, n_0}(\lambda'') \Sigma_{n, n_0}(\lambda'') \left[ V_{n, n_0}(\lambda'')^*U \right] \Sigma V^*~,\]
and let $F_1 = \Sigma V^* F_0$. If $\| \widetilde\Phi_n(\lambda'') v_0 \| \leq e^{-n\gamma}$ for some $v_0 \in S(F_0)$, then
\begin{equation}\label{eq:ev1} \| \Sigma_{n, n_0}(\lambda'') \left[ V_{n, n_0}(\lambda'')^*U\right] v_1 \| \leq e^{-n\gamma + C_1} \end{equation}
 for $v_1 = \Sigma V^* v_0 / \| \Sigma V^* v_0 \| \in S(F_1)$. Note that $\left[ V_{n, n_0}(\lambda'')^*U\right]$ is distributed uniformly on the compact symplectic group, and therefore its action on any fixed vector on the sphere is distributed uniformly on the sphere.
On the event $\Omega_{n,\epsilon}(\lambda)$, the numbers $a_j = \frac{1}{n} \log s_j(\Phi_{n_0, n}(\lambda))$ satisfy
\[ a_{2W+1-j} = - a_j~, \quad \sum_{i=1}^j a_i  \leq (1 - n_0/n) \sum_{i=1}^j \gamma_i(\lambda) + 2\epsilon j  \leq \sum_{i=1}^j \gamma_i(\lambda) + 2\epsilon W \quad (1 \leq j \leq W)~. \]
Therefore
\[\begin{split} \sum_{j=1}^{W+1} \left(\gamma - a_j\right)_+ &\geq 2 \gamma +\sum_{j=1}^{W-1} (\gamma - a_j)_+
    \\
& \geq 2 \gamma + (W-1) \big(\gamma - \frac{1}{W-1} \sum_{j=1}^{W-1} a_j\big)_+
\geq 2 \gamma + \big((W-1)\gamma - \sum_{j=1}^{W-1} \gamma_j(\lambda)\big)_+ - 2 \epsilon W~,\end{split}\]
whence
\[\sum_{j=1}^{W+1} (\gamma - \frac{C_1}{n} - a_j)_+ \geq 2 \gamma + \big((W-1)\gamma - \sum_{j=1}^{W-1} \gamma_j(\lambda)\big)_+ - 2 \epsilon W - \frac{2C_1W}{n} \]
According to Lemma~\ref{l:geom},
\[ \mathbb P\left\{ \text{(\ref{eq:ev1})} \text{ and }  \omega \in \Omega_{n,\epsilon}(\lambda) \right\} \leq
C_2 \exp\left\{ - n \left[ 2 \gamma + \big((W-1)\gamma - \sum_{j=1}^{W-1} \gamma_j(\lambda)\big)_+ - 2 \epsilon W \right] \right\}~, \]
as claimed in (\ref{eq:singlelam-1'}).
\end{proof}

Now we combine Claim~\ref{cl:2}  with Claim~\ref{cl:1} (applied to $\gamma - \log2 /n$ in place of $\gamma$) and (\ref{eq:p-omega}), and use the Fubini theorem:
\[\begin{split} &\mathbb P \left\{ \Fast_{n,\epsilon}(\gamma, \lambda) \neq \varnothing \right\} \\
&\quad \leq (1 - \mathbb P (\Omega_{n,\epsilon}(\lambda)) + 2 C r_\epsilon(\lambda) \exp\left\{ - n \left[ 2 \gamma + ((W-1)\gamma - \sum_{j=1}^{W-1} \gamma_j(\lambda))_+ - 2 \epsilon W \right] \right\}
	\eta^{-1} \\
&\quad \leq C n^3 e^{-cn} + C  n \exp\left\{ - n \left[ - \gamma_1(\lambda)  + \gamma + ((W-1)\gamma - \sum_{j=1}^{W-1} \gamma_j(\lambda))_+ - 4 \epsilon W \right] \right\}~.
\end{split}\]
For $\epsilon = \frac{1}{100 W} \min(\gamma - \gamma_{*,1}, 1)$, this expression tends to zero exponentially with $n$, thus concluding the proof of (\ref{eq:est-fast-1}) and of  Theorem~\ref{thm:1}.\qed


\section{Proof of Theorem~\ref{thm:2}}\label{s:2}

Let $\gamma > \frac{1}{3} (2\gamma_1(\lambda)+\gamma_2(\lambda))$, and let $\epsilon = \frac{1}{100} \min(\gamma_1(\lambda) - \gamma_2(\lambda), \gamma_2(\lambda), 1)$. We keep the notation $F_0$ (space of initial conditions, (\ref{eq:def-F0})), $\Omega_{n,\epsilon}(\lambda)$ (the event on which the products of singular values admit an upper bound, (\ref{eq:def-omega})), and $\Fast_{n,\epsilon}(\gamma, \lambda)$ (the set of energies $\lambda'$ in the vicinity of $\lambda$ for which there is a fast-decaying solution, (\ref{eq:def-fast})) from the previous section. Similarly to the previous section, or goal is to prove (\ref{eq:est-fast-1}), i.e.\ that $\Fast_{n,\epsilon}(\gamma, \lambda)$ is empty outside an event of exponentially small probability.

Denote by $u_j(\lambda')$ ($j=1,2,3,4$) the right singular vectors of $\Phi_n(\lambda')$ (i.e.\ the eigenvectors of $\Phi_n(\lambda')^* \Phi_n(\lambda')$; the choice of the direction of the vectors will be specified later), and by $P_{F_0}$ -- the orthogonal projection onto $F_0$. Let
\begin{align}
& \eta = \frac{1}{n} \exp(-n(\gamma - \gamma_2(\lambda)))~, \\
&A^+ = \left\{ \lambda'' \in (\lambda - r_\epsilon(\lambda), \lambda + r_\epsilon(\lambda)) \, : \, \| P_{F_0} u_1(\lambda'') \| \leq C \exp(-n(2\gamma-\gamma_1-\gamma_2-4\epsilon))\right\}~, \label{eq:def-A+}\end{align}
where $C>0$ will be specified shortly. The required estimate (\ref{eq:est-fast-1}) follows from (\ref{eq:p-omega}) and the following two ingredients: a propagation estimate
\begin{equation}\text{on $\Omega_{n,\epsilon}(\lambda)$}: \Big[ \lambda' \in A \overset{\text{def}}{=} \Fast_{n,\epsilon}(\gamma, \lambda)~, \,\, |\lambda'' - \lambda'| < \eta~, \,\,
|\lambda'' - \lambda| < r_\epsilon(\lambda) \Big] \Longrightarrow \lambda'' \in A^+ \label{eq:propag-2}\end{equation}
which replaces Claim~\ref{cl:2}, and the single-energy bound
\begin{equation}|\lambda'' - \lambda| < r_\epsilon(\lambda) \Longrightarrow \mathbb{P} \left\{ \lambda'' \in A^+\right\}  \leq C' e^{-\epsilon n} \eta \label{eq:est-2}
\end{equation}
which replaces Claim~\ref{cl:1}.

To prove (\ref{eq:propag-2}), we first observe that $\lambda'\in A$ implies that $\frac1n \log s_1(\Phi_n(\lambda')) \geq \gamma$, and hence on  $\Omega_{n,\epsilon}(\lambda)$
\begin{equation}\label{eq:ubd-emp}
\begin{split}
\frac1n \log s_2(\Phi_n(\lambda')) &= \frac1n \big(\log s_1(\Phi_n(\lambda'))  + \log s_2(\Phi_n(\lambda'))\big)- \frac1n \log s_1(\Phi_n(\lambda')) \\
&\leq \gamma_1(\lambda) + \gamma_2(\lambda) - \gamma + 4\epsilon~.
\end{split}\end{equation}
Further, $\lambda' \in A$ implies that there exists $v \in F_0$ such that for $j=1,2,3$
\[ |\langle v, u_j(\lambda') \rangle | \leq \frac{\exp(- n \gamma)}{s_j(\Phi_n(\lambda'))} \leq
\exp(- n \gamma) s_2(\Phi_n(\lambda')) \leq \exp(- n (2\gamma - \gamma_1(\lambda) - \gamma_2(\lambda) - 4\epsilon) )~.\]
These inequalities imply that (for the appropriate choice of signs)
\[ \| v - u_4(\lambda') \| \leq C_1 \exp(- n (2\gamma - \gamma_1(\lambda) - \gamma_2(\lambda) - 4\epsilon) )~.\]
Now we use the symplectic rotation $J = \left( \begin{array}{cc} 0 & - \mathbbm{1} \\ \mathbbm{1} & 0 \end{array} \right)$. The matrix $\Phi_n(\lambda')$ is symplectic, hence (up to sign) $J u_4(\lambda') = u_1(\lambda')$. Thus
\[ \| J v - u_1(\lambda') \| \leq C_1 \exp(- n (2\gamma - \gamma_1(\lambda) - \gamma_2(\lambda) - 4\epsilon) )~.\]
On the other hand, $F_0 \subset \mathbb R^{2W}$ is a Lagrangian subspace (i.e.\ $F_0 = (J F_0)^\perp$), hence $J v \perp F_0$. Consequently,
\begin{equation}\label{eq:proj-lam'} \| P_{F_0} u_1(\lambda') \|\leq  C_1 \exp(- n (2\gamma - \gamma_1(\lambda) - \gamma_2(\lambda) - 4\epsilon) )~.\end{equation}
To complete the proof of (\ref{eq:propag-2}), we need to show that the estimate (\ref{eq:proj-lam'}) does not deteriorate too fast as we vary $\lambda'$. If $|\lambda'' - \lambda'| \leq \eta$ and $|\lambda'' - \lambda| \leq r_\epsilon(\lambda)$, we have (still on $\Omega_{n,\epsilon}(\lambda)$, cf.\ (\ref{eq:16.5})):
\[ \| \Phi(\lambda'') - \Phi(\lambda') \| \leq \eta n \exp(n(\gamma_1(\lambda)+4\epsilon))~, \]
whence by Wedin's perturbation bound for singular vectors \cite{Wed}
\begin{equation}\label{eq:pert-u1}\begin{split} \| u_1(\lambda'') - u_1(\lambda')\| &\leq C_2 \frac{\eta n e^{n(\gamma_1(\lambda)+4\epsilon)}}{s_1(\Phi_n(\lambda')) - s_2(\Phi_n(\lambda'))} \\
&\leq \frac{2C_2 \eta n e^{n(\gamma_1(\lambda)+4\epsilon)}}{e^{n\gamma}} = 2C_2 \exp(-n(2\gamma - \gamma_1(\lambda) - \gamma_2(\lambda) - 4\epsilon ) ) \end{split}~.\end{equation}
On the second step we used that $s_1(\Phi_n(\lambda')) - s_2(\Phi_n(\lambda')) \geq \frac{1}{2} e^{n \gamma}$. This estimate holds (for sufficiently large $n$) since
\[ s_1(\Phi_n(\lambda')) \geq e^{n\gamma}~, \quad s_1(\Phi_n(\lambda)) s_2(\Phi_n(\lambda)) \leq e^{n(\gamma_1(\lambda) + \gamma_2(\lambda) + 4 \epsilon)}~, \]
whereas
\[\gamma > \frac{2\gamma_1(\lambda) + \gamma_2(\lambda)}{3}> \frac{\gamma_1(\lambda) + \gamma_2(\lambda)}{2}~.\]
From (\ref{eq:proj-lam'}) and (\ref{eq:pert-u1}) we obtain that $\lambda'' \in A^+$, provided that we set $C = C_1 + 2C_2$ in (\ref{eq:def-A+}). This concludes the proof of (\ref{eq:propag-2}).

To prove (\ref{eq:est-2}), we use once again that if $U$ is uniformly distributed on the compact symplectic group $\Sp(2W, \mathbb R) \cap \operatorname{SO}(2W, \mathbb R)$, then each column of $U$ is uniformly distributed on the unit sphere. Thus, according to Lemma~\ref{l:smooth}, the probability density of $u_1(\lambda'')$ with respect to the Haar measure on $S(\mathbb R^{2W})$  is bounded uniformly in $n \geq n_0$. Hence by (\ref{eq:archimedes})
\[ \begin{split} \mathbb P \left\{ \| P_{F_0} u_1(\lambda'') \| \leq C \exp(-n(2\gamma-\gamma_1-\gamma_2-4\epsilon)) \right\} &\leq C_4 \exp(-2n(2\gamma-\gamma_1-\gamma_2-4\epsilon)) \\
&\leq C_5 \exp(-\epsilon n) \eta~.\end{split}\]
This concludes the proof of (\ref{eq:est-2}) and of the theorem. \qed


\section{On the uniform convergence to the Lyapunov exponents}\label{s:concl}

A better understanding of  the  deviations of $\frac{1}{n} \log s_j(\Phi_n(\lambda))$ from their limiting values $\gamma_j$ would allow us to strengthen the conclusion of Theorems~\ref{thm:1} and \ref{thm:2}, possibly up to the conjectured $\gamma_*^+ \overset{\text{conj}}{=} \gamma_1$, as we now discuss.

Recall the following result from \cite{G75,G80a} pertaining to $W=1$: with probability one, the set  $\Lambda_{\frac12}$, where
\[ \Lambda_{\tau} = \left\{ \lambda \in \mathbb R\, :  \, \liminf_{n \to \infty} \frac{1}{n} \log \|\Phi_n(\lambda) \| \leq \tau\gamma_1(\lambda)\right\}~, \quad \tau \in [0,1]~, \]
is dense in $\sigma(H)$. Subsequently, it was found that also the (possibly) smaller set $\Lambda_{0}$
is almost surely dense in $\sigma(H)$. Recently, a  general framework encompassing and generalising these results was developed by Gorodetski and Kleptsyn \cite{GorKl}, who also provided  detailed information on the structure of the exceptional sets $\Lambda_\tau$, and showed that
\begin{equation}\label{eq:gorkl}
\mathbb P \left\{ \forall  \lambda \in \mathbb R\, :  \, \limsup_{n \to \infty} \frac{1}{n} \log \|\Phi_n(\lambda) \| = \gamma_1(\lambda)\right\} = 1~. \end{equation}

We are not aware of a published reference discussing the extension of this problem for $W > 1$. However, it is plausible that the arguments developed in the aforementioned works could yield that
\[ \Lambda_{0}^{(W)} = \left\{ \lambda \in \mathbb R\, : \, \liminf_{n \to \infty} \frac{1}{n} \log s_W(\Phi_n(\lambda)) =0 \right\}  \]
is dense in $\sigma(H)$. It is not clear to us what would be the right counterpart of this statement for $1 \leq j \leq W-1$. If the higher  exponents would exhibit regular behaviour, i.e.\
\begin{equation}\label{eq:??}
\text{if it were true that \,} \mathbb P \left\{ \forall \lambda \in \sigma(H) \quad \lim_{n \to \infty} \frac{1}{n} \log s_j(\Phi_n(\lambda)) = \gamma_j(\lambda) \right\} =1~, \quad 1 \leq j \leq W-1~, \end{equation}
one could significantly improve the results of the current paper: the argument of Theorem~\ref{thm:1} would yield $\gamma_*^+ \leq \gamma_{*,2}$, where $\gamma_{*,2}$ is the solution of
\[ \gamma + \sum_{j=1}^{W-1} \left(  \gamma - \gamma_j\right)_+ = \gamma_1~, \]
whereas  the argument of Theorem~\ref{thm:2} would establish the optimal bound $\gamma_*^+ = \gamma_W$ (for arbitrary $W$, cf.\ the proof of Theorem~\ref{thm:3} below). If (\ref{eq:??}) is false, it would be helpful to understand
\begin{equation}\label{eq:???}
\text{is it true that \,} \mathbb P \left\{ \forall \lambda \in \sigma(H) \quad \limsup_{n \to \infty} \frac{1}{n} \log s_j(\Phi_n(\lambda)) \leq \gamma_j(\lambda) \right\} =1~, \quad 1 \leq j \leq W~. \end{equation}
Following Craig and Simon \cite{CS} (cf.\  Lemma~\ref{l:upperbd}), note  that (\ref{eq:???}) holds (unconditionally) for $j=1$. Also (according to the same lemma) (\ref{eq:??}) would imply (\ref{eq:???}).

\medskip
In this section, we prove the following extension of (\ref{eq:gorkl}) to $2W$-dimensional cocycles. We confine ourselves to the setting of transfer matrices, which is used in the proof of Theorem~\ref{thm:3}. Denote
\[ \operatorname{Dev}_{n}(\lambda) = \max_{1 \leq j \leq W} \left|\frac1n \log s_j(\Phi_n(\lambda)) - \gamma_j(\lambda)\right|~.\]
\begin{prop}\label{p:gorkl} Assume that $V(n)$ satisfy (\ref{eq:model}), and that (\ref{eq:zcond}) holds for every $\lambda \in [a,b]$. Then for any $\epsilon > 0$ there exist $C> 0$ and $c > 0$ such that
\begin{equation}
\mathbb P \left\{ \sup_{\lambda \in [a,b]} \min\left(\operatorname{Dev}_n(\lambda), \operatorname{Dev}_{n^2}(\lambda)\right) \geq \epsilon  \right\}  \leq C e^{-cn}~.
\end{equation}
In particular,
\[ \mathbb P \left\{ \sup_{\lambda \in [a,b]} \liminf_{n \to \infty} \operatorname{Dev}_n(\lambda)= 0 \right\} = 1~. \]
\end{prop}
\begin{rmk} Here $n^2$ can be replaced with any function tending to infinity faster than linearly.
\end{rmk}

\begin{proof}[Proof of Proposition~\ref{p:gorkl}] Fix $\lambda \in \mathbb R$; let $\epsilon > 0$, and choose $r_\epsilon(\lambda)$ as in (\ref{eq:def-r}). It will suffice to show that there exist $C, c$ such that
\begin{equation}\label{eq:need-gk}
\mathbb P \left\{ \sup_{|\lambda' - \lambda| < r_\epsilon(\lambda)} \min\left(d_n(\lambda'), d_{n^2}(\lambda')\right) \geq 10 W \epsilon  \right\}  \leq C e^{-cn}~,
\end{equation}
where
\begin{equation}\label{eq:dn}d_n(\lambda') = \max_{1 \leq j \leq n} \left|\frac1n \log s_j(\Phi_n(\lambda')) - \gamma_j(\lambda)\right|~.\end{equation}
By the existence of a fractional moment (\ref{eq:model}) and the Chebyshev inequality, one can choose $\kappa > 0$ such that
\[ \mathbb P (\Omega^{(1)}_n) \geq 1 - e^{-n}~, \quad \Omega^{(1)}_n = \left\{ \sup_{|\lambda' - \lambda| < r_\epsilon(\lambda)} \|\Phi_n(\lambda')\| \leq e^{\kappa n} \right\}~.\]
On $\Omega^{(1)}_n$,
\[ \left|\log s_j(\Phi_{n^2}(\lambda'))  - \log s_j(\Phi_{n^2, n}(\lambda')) \right| \leq  \kappa n~,\]
therefore for sufficiently large $n$
\[ d_{n^2}(\lambda') \leq \epsilon + \tilde d_{n^2}(\lambda')~, \quad \tilde d_{n^2}(\lambda')  = \max_{1 \leq j \leq n} \left|\frac1{n^2-n} \log s_j(\Phi_{n^2,n}(\lambda')) - \gamma_j(\lambda)\right|~. \]
Here $\tilde d_{n^2}(\cdot)$ is independent of $d_n(\cdot)$.  Also recall from Lemma~\ref{l:upperbd} that $\mathbb P(\Omega_n^{(2)}) \geq 1 - C e^{-cn}$, where
\[ \Omega_n^{(2)} = \left\{ \forall  \lambda' \in  (\lambda - r_\epsilon(\lambda),   \lambda + r_\epsilon(\lambda))~, \,  1 \leq j \leq n\,\, : \,\, \frac{1}{n} \sum_{i=1}^j \log s_i(\Phi_n(\lambda')) \leq \sum_{i=1}^j \gamma_i(\lambda) + 2j \epsilon \right\} ~.\]
Now, Lemma~\ref{l:largedev} implies that for each $\lambda'\in  (\lambda - r_\epsilon(\lambda), \lambda + r_\epsilon(\lambda))$
\begin{equation}\label{eq:fromldbd}\mathbb P \left\{ |\frac{1}{n} \log |||\Phi_n(\lambda')||| - \gamma_1(\lambda) | \geq 2\epsilon  \right\} \leq C \exp(-c n)~, \end{equation}
where $|||X||| = \max_{\alpha,\beta} |X_{\alpha,\beta}|$ (the $\ell_1$ to $\ell_\infty$ norm).
Note that each matrix entry of $\Phi_n(\lambda')$ is a polynomial in $\lambda'$ of degree $n$, therefore the set $A_n^{(1)}$ of $\lambda' \in  (\lambda - r_\epsilon(\lambda),\lambda + r_\epsilon(\lambda))$ for which $|\frac{1}{n} \log |||\Phi_n(\lambda')||| - \gamma_1(\lambda) | \geq 2\epsilon$ is a union of at most $W^2 n$ intervals. Applying the same argument to the wedge powers $\Phi_n(\lambda)^{\wedge j}$, we construct the sets $A_n^{(2)}, \cdots, A_n^{(W)}$ such that $A_n^{(j)}$ is a union of at most $C_j(W) n$ intervals (where $C_j(W)$ may depend only on $j$ and $W$),
\begin{equation}\label{eq:prob-An} \mathbb P \left\{ \lambda' \in A_n^{(j)}\right\} \leq C e^{-cn}~, \end{equation}
and
\[ \lambda' \in (\lambda - r_\epsilon(\lambda), \lambda + r_\epsilon(\lambda)) \setminus A_n^{(j)}
\Longrightarrow \frac1n \sum_{i=1}^j \log s_i(\Phi_n(\lambda')) \geq \sum_{i=1}^j \gamma_i(\lambda) - 2j \epsilon~.\]
We also construct similar sets $A_{n^2,n}^{(j)}$ corresponding to $\Phi_{n^2,n}(\lambda')$, and let
\[ A_n = \bigcup_{j=1}^W A_n^{(j)}~, \quad A_{n^2,n} = \bigcup_{j=1}^W A_{n^2,n}^{(j)}~. \]
The set $A_n$ is a union of $\leq C(W) n$ intervals, whereas $A_{n^2,n}$ is a union of $\leq C(W) n^2$ intervals. If these two sets intersect, than either one of the edges of the intervals comprising $A_n$ lies in $A_{n^2,n}$, or vice versa. Invoking (\ref{eq:prob-An}), we see that
\[\mathbb P (\Omega_n^{(3)}) \geq 1 - C e^{-cn}~, \quad \text{where} \quad \Omega_n^{(3)} = \left\{ A_n \cap A_{n^2, n} = \varnothing \right\} ~. \]
Observe that on $\Omega_n^{(1)} \cap \Omega_n^{(2)}  \cap \Omega_n^{(3)}$, for each $\lambda'$, either $\frac{1}{n} \log s_j(\Phi_n(\lambda'))$ is close to $\gamma_j(\lambda)$ for all $j$, or this holds true for $\frac{1}{n^2} \log s_1(\Phi_{n^2}(\lambda'))$. This concludes the proof of the proposition.
\end{proof}


\section{Proof of Theorem~\ref{thm:3}}\label{s:3}
We keep the notation from the previous sections. Let $\gamma > \gamma_W(\lambda)$, and let $\epsilon = \frac{1}{100 W^2} (\gamma - \gamma_W(\lambda))$. It suffices to show that
\begin{equation}\label{eq:cl-3}
\mathbb P \left\{ \Fast_{n,\epsilon}(\gamma, \lambda) \cap \Fast_{n^2,\epsilon}(\gamma, \lambda) \neq \varnothing  \right\}
\leq C e^{-cn}~.
\end{equation}
To keep the notation consistent with the previous sections, it will be convenient to rely on the estimate (\ref{eq:need-gk}) rather than on the conclusion of Proposition~\ref{p:gorkl}. Denote
\[ \Reg_{n,\epsilon}(\lambda) = \left\{ \lambda' \in (\lambda - r_\epsilon(\lambda),\lambda + r_\epsilon(\lambda)) \, : \,
d_n(\lambda') < 10W\epsilon\right\}\]
where $d_n$ are as in (\ref{eq:dn}). From (\ref{eq:need-gk}),
\[ \mathbb P\left\{ \Reg_{n,\epsilon}(\lambda) \cup \Reg_{n^2,\epsilon}(\lambda) =  (\lambda - r_\epsilon(\lambda),\lambda + r_\epsilon(\lambda))\right\} \geq 1 - C e^{-cn}~. \]
Therefore (\ref{eq:cl-3}) and the  theorem are implied by (\ref{eq:p-omega}) and the following estimate:
\begin{equation}\label{eq:est-3}\mathbb P \left\{\Fast_{n,\epsilon}(\gamma, \lambda) \cap \Reg_{n,\epsilon}(\lambda)  \neq \varnothing; \omega \in \Omega_{n,\epsilon}(\lambda) \right\} \leq C e^{-cn}~.
\end{equation}
The proof of (\ref{eq:est-3}) is similar to the argument in Section~\ref{s:2}. Denote
\[ A = \Fast_{n,\epsilon}(\gamma, \lambda) \cap \Reg_{n,\epsilon}(\lambda)~, \quad \eta = \frac1n \exp(- (\gamma - \gamma_W(\lambda) + 20W^2 \epsilon )n)~, \]
and let $A^+$ be the set of $\lambda'' \in (\lambda - r_\epsilon(\lambda), \lambda + r_\epsilon(\lambda))$ for which there exists
\begin{equation}\label{eq:w-3}w \in S(\operatorname{span}(u_1(\lambda''), \cdots, u_{W-1}(\lambda'')))~, \,\,  \| P_{F_0} w \| \leq C \exp(-n (\gamma - \gamma_W(\lambda) - 10 W \epsilon))~, \end{equation}
where $C>0$ will be specified later.
We claim that on $\Omega_{n,\epsilon}(\lambda)$ we have the propagation estimate
\begin{equation}\label{eq:propag-3}\lambda' \in A~, \,\, |\lambda'' - \lambda'| < \eta~,\,\,|\lambda'' - \lambda | < r_\epsilon(\lambda) \Longrightarrow \lambda'' \in A^+~,\end{equation}
and that for each $\lambda'' \in (\lambda - r_\epsilon(\lambda), \lambda + r_\epsilon(\lambda))$
\begin{equation}\label{eq:prob-3}
\mathbb P \left\{ \lambda'' \in A^+ \right\} \leq C' \exp(-2n (\gamma - \gamma_W(\lambda) - 10 W \epsilon))
\end{equation}
These two claims imply (\ref{eq:est-3}) and thus conclude the proof of  the theorem.

To prove (\ref{eq:propag-3}), we observe that if $\lambda' \in A$, there exists $v \in S(F_0)$ such that for all $1 \leq j \leq W+1$
\[  | \langle v, u_j(\lambda') \rangle| \leq \frac{\exp(-n \gamma)}{s_j(\Phi_n(\lambda'))}
\leq \exp(-n(\gamma - \gamma_W(\lambda) - 10 W \epsilon))~, \]
where $u_j(\lambda')$ is the $j$-th right singular vector of $\Phi_n(\lambda')$, and thus there exists $\theta \in S(\mathbb R^{W-1})$ such that
\[ \| v - \sum_{j=1}^{W-1} \theta_j u_{2W+1-j}(\lambda')\| \leq C_1 \exp(-n(\gamma - \gamma_W(\lambda) - 10 W \epsilon))~.\]
Now let $w = Jv$, where $J$ is the symplectic rotation. Then $w \perp F_0$, and
\begin{equation}\label{eq:w-prop}\| w - \sum_{j=1}^{W-1} \theta_j u_{j}(\lambda')\| \leq C_1 \exp(-n(\gamma - \gamma_W(\lambda) - 10 W \epsilon))~. \end{equation}
Applying Wedin's bound to the $j$-th wedge power of $\Phi_n(\lambda')$, we have:
\[\begin{split}
&\|u_1(\lambda'') \wedge u_2(\lambda'') \wedge \cdots \wedge u_j(\lambda'') -  u_1(\lambda') \wedge u_2(\lambda') \wedge \cdots \wedge u_j(\lambda') \| \\
&\quad \leq \frac{C_2 \eta n e^{(\gamma_1(\lambda) + \cdots + \gamma_j(\lambda) + 4 W \epsilon)n}}{e^{(\gamma_1(\lambda) + \cdots + \gamma_j(\lambda) -10W^2 \epsilon)n}} \leq C_2 \eta n  e^{12 W^2 \epsilon n}
\leq C_2 e^{- n(\gamma - \gamma_W(\lambda) - 10 W \epsilon)}~,\end{split}\]
and consequently
\[ \| u_j(\lambda'') - u_j(\lambda')\| \leq C_3 e^{- n(\gamma - \gamma_W(\lambda) - 10 W \epsilon)}~.\]
This and (\ref{eq:w-prop}) implies
\begin{equation}\label{eq:w-prop'}\| w - \sum_{j=1}^{W-1} \theta_j u_{j}(\lambda'')\| \leq C_4 \exp(-n(\gamma - \gamma_W(\lambda) - 10 W \epsilon))~, \end{equation}
i.e.\ $\lambda '' \in A^+$ (if $C$ in (\ref{eq:w-3}) is chosen appropriately), as claimed in (\ref{eq:propag-3}).

Now we prove (\ref{eq:prob-3}). If
\begin{equation}\label{eq:prop-alpha}
\| P_{F_0} \sum_{j=1}^{W-1} \theta_j u_{j}(\lambda'') \| \leq C \exp(-n(\gamma - \gamma_W(\lambda) - 10 W \epsilon)) \end{equation}
for a certain $\theta \in S(\mathbb R^{W-1})$, then
\[\| P_{F_0} \sum_{j=1}^{W-1} \theta_j' u_{j}(\lambda'') \| \leq 2C \exp(-n(\gamma - \gamma_W(\lambda) - 10 W \epsilon)) \]
for all $\theta'$ in a neighbourhood of $\theta$ on $S(\mathbb R^{W-1})$; the $(W-2)$-dimensional volume of this neighbourhood is bounded from below by
\[ c \exp(-(W-2) n(\gamma - \gamma_W(\lambda) - 10 W \epsilon))~. \]
 On the other hand, (\ref{eq:archimedes}) implies that for each $\theta\ \in S(\mathbb R^{W-1})$
\[ \mathbb P \left\{ \| P_{F_0} \sum_{j=1}^{W-1} \theta'_j u_{j}(\lambda'') \| \leq 2C \exp(-n(\gamma - \gamma_W(\lambda) - 10 W \epsilon)) \right\} \leq C' \exp(-n W(\gamma - \gamma_W(\lambda) - 10 W \epsilon))~. \]
Therefore the probability that there exists $\theta$ satisfying (\ref{eq:prop-alpha})  is at most
\[ \frac{C'}{c} \exp(-2n (\gamma - \gamma_W(\lambda) - 10 W \epsilon))~,\]
as claimed. This concludes the proof of (\ref{eq:prob-3}) and of the theorem.
\qed

\paragraph{Acknowledgement} We are grateful to Yotam Hendel and Itay Glazer for explaining us the argument used in the proof of Lemma~\ref{l:smooth}, and to Fulvio Ricci for helpful comments on the work \cite{RS}.

\end{document}